%% file: main_submitted.tex
\documentclass{article}

\usepackage{microtype}
\usepackage{graphicx}
\usepackage{booktabs} 

\usepackage{hyperref}

\usepackage[accepted]{icml2025}

\input{math_commands.tex}

\usepackage{hyperref}
\usepackage{url}
\usepackage{graphicx}
\usepackage{booktabs}
\usepackage{siunitx}
\usepackage{adjustbox}

\usepackage[utf8]{inputenc} 
\usepackage[T1]{fontenc}    
\usepackage{hyperref}       
\usepackage{url}            
\usepackage{booktabs}       
\usepackage{amsfonts}       
\usepackage{nicefrac}       
\usepackage{microtype}      
\usepackage{xcolor}         

\usepackage{dcolumn}
\newcolumntype{d}[1]{D{.}{.}{#1}}

\usepackage{amsthm}
\usepackage{amsmath}
\usepackage{amssymb}
\usepackage{breqn}
\usepackage{dsfont}
\usepackage{mathtools} 
\usepackage{subfig}
\usepackage{xspace}

\usepackage{array}
\newcolumntype{L}[1]{>{\raggedright\let\newline\\\arraybackslash\hspace{0pt}}m{#1}}
\newcolumntype{C}[1]{>{\centering\let\newline\\\arraybackslash\hspace{0pt}}m{#1}}
\newcolumntype{R}[1]{>{\raggedleft\let\newline\\\arraybackslash\hspace{0pt}}m{#1}}

\usepackage{wasysym}

\DeclareMathOperator*{\topk}{topk}

\newcommand{\eat}[1]{}

\DeclareFontFamily{U}{mathx}{}
\DeclareFontShape{U}{mathx}{m}{n}{<-> mathx10}{}
\DeclareSymbolFont{mathx}{U}{mathx}{m}{n}
\DeclareMathAccent{\widehat}{0}{mathx}{"70}
\DeclareMathAccent{\widecheck}{0}{mathx}{"71}

\newtheorem{theorem}{Theorem}
\newtheorem{oq}{Open Question}
\usepackage{enumitem}

\begin{document}

\icmltitlerunning{Position: Certified Robustness Does Not (Yet) Imply Model Security}

\twocolumn[
\icmltitle{Position: Certified Robustness Does Not (Yet) Imply Model Security}

\icmlsetsymbol{equal}{*}

\begin{icmlauthorlist}
\icmlauthor{Andrew C. Cullen}{yyy}
\icmlauthor{Paul Montague}{dst}
\icmlauthor{Sarah Erfani}{yyy}
\icmlauthor{Benjamin I. P. Rubinstein}{yyy}
\end{icmlauthorlist}

\icmlaffiliation{yyy}{School of Computing and Information Systems, University of Melbourne, Australia}
\icmlaffiliation{dst}{DST Group, Adelaide, Australia}

\icmlcorrespondingauthor{Andrew C. Cullen}{andrew.cullen@unimelb.edu.au}

\icmlkeywords{Machine Learning, ICML}

\vskip 0.3in
]

\printAffiliationsAndNotice{}  

\begin{abstract}

While certified robustness is widely promoted as a solution to adversarial examples in Artificial Intelligence systems, significant challenges remain before these techniques can be meaningfully deployed in real-world applications. We identify critical gaps in current research, including the paradox of detection without distinction, the lack of clear criteria for practitioners to evaluate certification schemes, and the potential security risks arising from users' expectations surrounding ``guaranteed" robustness claims. These create an alignment issue between how certifications are presented and perceived, relative to their actual capabilities. This position paper is a call to arms for the certification research community, proposing concrete steps to address these fundamental challenges and advance the field toward practical applicability.

\end{abstract}

\section{Introduction}\label{sec:introduction}

A known property of learned models like neural networks is that they can have their outputs changed through semantically indistinguishable changes to their inputs~\cite{biggio2013evasion,szegedy2013intriguing}. The risk associated with these manipulated samples---known as \emph{adversarial examples}---is heightened by both the confidence ascribed by models to these samples~\cite{kumar2020adversarial}, and how simple they are to construct, with mechanisms typically relying upon the same gradient descent style processes that are used in model training~\cite{papernot2017practical, carlini2017towards}. As Artificial Intelligence (AI) increasingly permeates both interpersonal and business interactions, these adversarial examples have the potential to impact the security of real world systems~\cite{ibitoye2019threat, finlayson2019adversarial, albert2020politics, liu2025multi}. 

In response to these security concerns, significant research effort has been devoted to what are known as \emph{adversarial defenses}, which are designed to minimize or mitigate specific attacks~\cite{chakraborty2018adversarial}. However, it is crucial to note that these defenses are rarely more than \emph{best response} strategies to particular attack vectors, as deployed mitigations can often be defeated by identifying a single, undefended vector~\cite{perolat2018playing}. This makes AI security into an expensive, reactive process that requires constant vigilance by model deployers. 

Certified defenses eschew this best response paradigm by guaranteeing the absence of a family of potential attacks, rather than any one attack~\cite{lecuyer2019certified, li2018certified}. These families are typically parameterized by $\ell_p$ inputs, with the certification guaranteeing the absence of adversarial examples within a calculated region. While the literature has primarily focused upon classifiers, recent works have begun to explore extensions of certified defenses to other problem spaces, including reinforcement learning and regression~\cite{han2018reinforcement, lutjens2020certified, hammoudeh2022reducing, liu2023enhancing, rekavandi2024certified, liu2025Fox}. 

For all the promise of such systems, when it comes to the practical implications of these approaches, the devil is very much in the details. These techniques are presented as producing the distance to the \emph{nearest adversarial example}, which implies that the sample being certified is itself not already attacked. However, as we observed by~\citet{cullen2023exploiting}, a more precise framing is that: a certification bounds the distance to the \emph{nearest class changing example}. This distinction may appear minor, however, it is crucial for understanding the limitations of these techniques, as there is no guarantee that the class prediction is accurate. As such the certification could be the distance from either a clean sample to an adversarial instance, an adversarial instance to a different adversarial class, or an adversarial sample to a clean sample. Thus current certifications provide no information for distinguishing between clean and attacked samples, with certifications existing for both, and thus, there is no inherent security inferred by the certificate. Even the idea that a certification can be considered a measure for how much effort would be required to attack a particular sample does not hold, given the observation that certifications themselves can be exploited to guide adversarial attackers~\cite{cullen2023exploiting}. 

This then leads to a problematic alignment issue, in which there is a contradiction between how these systems are presented---as reliable, \emph{guarantees} against adversarial manipulation---and the practical reality of how they perform. And this difference is crucial for practical security, as it produces an attack surface between how the security of these systems may be perceived, relative to what it can deliver. It is for these reasons that we take the position that today's certified guarantees may provide more security theater than actual security, especially if they are providing a false-sense of security to users who are not fully across these technical nuances.

\subsection{Why This Position Paper?}

Adversarial examples present a clear and present danger to AI models, and the risks associated with their existence will only grow as models are more frequently integrated into systems where incentives exist for adversarial manipulation. While certified defenses have been presented as an incorruptible solution to adversarial examples, the current literature does not support their practical application. \textbf{This paper argues that current research into certified robustness is not aligned with the provision of model security, and may, in fact, harm security}. Our interrogation of this point is supported by:
\begin{enumerate}[noitemsep,topsep=0pt]
    \item Examining the gap between the \emph{ideal} of certifications and their practical implications. 
    \item Presenting best practice for aligning research, development, and deployment to tighten model security.
    \item Arguing that an application-driven approach is crucial for enhancing the impact of certifications.
\end{enumerate}

\section{Securing Against Adversarial Manipulation}

The performance dividends made possible by deploying AI has lead to its inclusion in a broad swathe of real world systems. However, these systems introduce new frontiers of risk, as these models are incredibly sensitive to being exploited by a motivated attacker. 

Within the AI security community, significant research interest has been placed upon norm-minimising $\ell_p$ evasion attacks~\cite{papernot2016limitations}, which attempt to induce a class change at test (or inference) time to minimize the $\ell_p$ distance between the original sample and its corresponding adversarial example. The appeal of this threat model is multifaceted: it affects a broad spectrum of systems (including classifiers and reinforcement learning); leads to easy-to-construct attacks; and correlates with our conceptual understanding of human- and machine-perceptibility~\cite{gilmer2018motivating}. 

Attacks in the form of adversarial examples can be considered as variants of gradient descent, involving finding a class flipping example $\mathbf{x}'$ that approximates the minima of
\begin{align}
    &\argmin_{\mathbf{x}' \in \mathcal{S}} \| \mathbf{x} - \mathbf{x}'\|_p \\
    &\text{s.t. } \argmax_{i \in \mathcal{K}} f_i(\mathbf{x}) \neq \argmax_{i \in \mathcal{K}} f_i(\mathbf{x}')\enspace, \nonumber
\end{align}
across some permissible space $\mathcal{S}$, which is typically the d-dimensional space $[0,1]^d$ for computer vision. This framing has led to a number of distinct evasion attacks, including PGD~\cite{madry2017towards}, Carlini-Wagner~\cite{carlini2017towards}, DeepFool~\cite{moosavi2016deepfool}, and AutoAttack~\cite{croce2020reliable}, and has been shown to have the potential to compromise real world systems~\cite{wu2020making, cullen2023predicting}. Similar mechanisms also can be deployed to attack models at training time, to either corrupt learning performance or embed deleterious behaviors into the model's outputs.

\subsection{Adversarial Defenses}

While early works suggested that techniques like model regularisation and weight decay may minimize the success rates of these attacks~\cite{kukavcka2017regularization}, these mitigations have been broadly shown to be ineffective~\cite{kurakin2016adversarial, athalye2018obfuscated}. Consequently, research has shifted to developing countermeasures, known as adversarial defenses. While these approaches have demonstrated more success, they also share a common weakness, in that they serve as responses to specific attacks, and do not typically provide resistance against alternative attack frameworks. This has led to a cyclical development process, where defenses are attacked, and new defenses are subsequently proposed to counter those attacks. An example of this is single step-attacks~\cite{goodfellow2014explaining} being mitigated by adversarial training, which led to the development of multi-stage attacks~\cite{kurakin2016adversarial}. These were in turn countered by defensive distillation~\cite{papernot2016distillation}, which has subsequently been attacked. This game of cat-and-mouse demonstrates that an attacker only needs to find an undefended vector to carry out their attack. Therefore, the adversarial resistance offered by a defense is, at best, limited when faced with a motivated attacker who can evade or exploit the deployed system~\cite{meng2017magnet,carlini2017magnet}.

\subsection{Certified Defenses}

In response to the inadequacy of adversarial defenses, the AI and security communities have developed certified defenses, which construct regions around samples in which it can be guaranteed that no adversarial example exists. Crucially, these guarantees are independent of the particular attack framework, and only make basic assumptions regarding the threat model associated with the attacker. 

Certification mechanisms eschew the reactive view of adversarial defenses in favor of proactively bounding the space within which adversarial examples can exist. In some mechanisms, this might be a $\mathbf{x}$-centered $p$-norm ball of radius $r$ defined as $B_p(\mathbf{x}, r)$, where $r$ is strictly less than
\begin{align}\label{eqn:r_definition}
    r^\star &= \inf\left\{ \|\mathbf{x}-\mathbf{x}'\|_p : \mathbf{x}'\in \mathcal{S}, F(\mathbf{x})\neq F(\mathbf{x}'), \right\} \\
    &\text{           where }
    F(\cdot) = \mathds{1} \left(\argmax_{i \in \mathcal{K}} f_i(\cdot)\right) \enspace.\nonumber
\end{align}
Here $\mathds{1}$ is a one-hot encoding of the predicted class in $\mathcal{K} = \{1, \ldots, K\}$. The size of $B_p(\mathbf{x}, r)$ can be considered a reliable proxy for both the \emph{detectability} of adversarial examples~\citep{gilmer2018motivating} and the \emph{cost} to the attacker~\citep{huang2011adversarial}. 

The construction of such bounds are typically approached through either exact or high-probability methods, with interval bound propagation (IBP) and convex relaxation~\cite{mirman2018differentiable, weng2018towards} being examples of the former, and randomized smoothing~\cite{lecuyer2019certified} thelatter. While high probability methods construct high-probability bounds on the existence of adversarial examples, exact methods construct bounds by propagating intervals through the model and tracing potential class changes. 

Exact approaches require significant changes to training processes and place limits on available model architectures. Moreover, these techniques impose a significant computational cost (in terms of time and GPU memory) that scales with model size, due to the need to propagate bounds through each layer~\cite{zhang2018efficient, xu2020automatic}, which typically requires the introduction of approximations to scale to model sizes of academic/industrial interest~\cite{gowal2018scalable, singh2018fast}. In contrast, randomized smoothing can be applied to almost any model architecture or training routine, with the only required changes occurring before the input layer and after the output layer. While randomized smoothing does require significant numbers of model samples to be evaluated, the computational time implications of this can be partially ameliorated, as the sampling process is embarrassingly parallel, making it a powerful alternative to bound-propagation style approaches~\cite{cohen2019certified}.

We must emphasize that the following definitions, and their associated techniques, are heavily aligned with evasion attacks in the interests of clarity. While some recent works have begun to consider constructing certifications against other attack frameworks, there still exist a broad range of attacks outside the aegis of evasion attacks, including backdoor attacks that embed deleterious behaviors that can be manipulated; model stealing attacks, where proprietary information is extracted from the model; check fraud, which forces the model to read a larger amount of data than what is written~\cite{papernot2016distillation} and more. 

\subsubsection{Randomized Smoothing}

The certifications constructed by randomized smoothing~\citep{lecuyer2019certified} are built around a Monte Carlo estimator of the expectation of a class prediction, where 
\begin{eqnarray}\label{eqn:expectations}
     \frac{1}{N} \sum_{j=1}^{N} F(\mathbf{X}_j) &\approx& \mathbb{E}_{\mathbf{X}}[F(\mathbf{X})]  \qquad \forall i \in \mathcal{K} \\
        \mathbf{X}_1, \ldots, \mathbf{X}_N, \mathbf{X} &\stackrel{i.i.d.}{\sim}& \mathbf{x} + \mathcal{N}(0, \sigma^2)\enspace. \nonumber
\end{eqnarray}
These expectations can be employed to provide guarantees of invariance under \emph{additive} perturbations. In forming this aggregated classification, the model is re-construed as a \emph{smoothed classifier}, which in turn is certified. Mechanisms for constructing such certifications include differential privacy~\citep{lecuyer2019certified,dwork2006calibrating}, R\'{e}nyi divergence \citep{li2018certified}, and parameterising worst-case behaviors \citep{cohen2019certified, salman2019provably, cullen2022double}. The latter of these approaches has proved the most performant, and yields certifications of the form
\begin{equation}\label{eqn:Cohen_Bound}
r = \frac{\sigma}{2} \left( \Phi^{-1}\left(\widecheck{E}_{0}[\mathbf{x}]\right) - \Phi^{-1}\left(\widehat{E}_{1}[\mathbf{x}]\right) \right)\enspace,
\end{equation}
where $\Phi^{-1}$ is the inverse normal CDF, 
$(E_0, E_1) = \topk\left(\left\{\mathbb{E}_{\mathbf{X}}[F(\mathbf{X})]\right\}, 2\right)$
, and $(\widecheck{E}_0, \widehat{E}_1)$ are the lower and upper confidence bounds of these quantities to some confidence level $\alpha$ \citep{goodman1965simultaneous}. 

\subsubsection{Interval Bound Propagation}

Conservative certificates upon the impact of norm-bounded perturbations can be constructed by way of either interval bound propagation (IBP) which propagates interval bounds through the model; or convex relaxation, which utilizes linear relaxation to construct bounding output polytopes over input bounded perturbations. In contrast to randomized smoothing, which constructs isotropic measures of $\ell_p$-robustness, interval bound propagation and its associated techniques attempt to propagate the potential influence of all possible perturbations through the model, producing an anisotropic measure of the potential response of a model to any potential perturbation~\cite{salman2019convex, mirman2018differentiable, weng2018towards, CROWN2018, zhang2018efficient, singh2019abstract, mohapatra2020towards}. Of these, IBP is more general, while convex relaxation typically provides tighter bounds~\cite{lyu2021towards}.  

Utilizing these techniques requires introducing an augmented loss function during training to promote tight output bounds \cite{xu2020automatic}. These schemes have also, until very recently, been heavily limited in the types of network architectures that they can successfully construct bounds through, with only recent works demonstrating an applicability to a nonlinear activation functions beyond ReLU~\cite{shi2023formal}. Moreover they both exhibit a time and memory complexity that makes them infeasible for complex model architectures or high-dimensional data~\cite{wang2021beta, chiang2020certified, levine2020randomized}.

\subsubsection{Global Lipschitz}

Global Lipschitz takes an alternative approach to constructing certifications, a point that they distinguish through the framing of local and global robustness. The guarantees provided by prior works, which can take the form
%
\begin{equation}
    \| \mathbf{x} - \mathbf{x}' \|_p \leq \epsilon \implies F(\mathbf{x}) = F(\mathbf{x}')
\end{equation}
are considered to be local properties, that relate $\mathbf{x}$ and $\epsilon$. Lipschitz based techniques instead attempt to construct their certifications in terms of \emph{global} robustness, where
\begin{equation}
    \forall \mathbf{x}_1, \mathbf{x}_2 : \| \mathbf{x}_1 - \mathbf{x}_2 \|_p \leq \epsilon \implies F(\mathbf{x}_1) \overset{\perp}{=}  F(\mathbf{x}_2)\enspace.
\end{equation}
Here $\perp$ is the marker for an \emph{abstained} class prediction, and $c_1 \overset{\perp}{=} c_2$ denotes that either $c_1 = \perp$, $c_2 = \perp$, or $c_1 = c_2$. In essence such a form of certification involves constructing a model that has not only an infinitesimally thin decision boundary, but a margin between the regions associated with each class, where $\epsilon$ then becomes the shortest $\ell_p$ distance to span the boundary. Several attempts have been made to use Lipschitz bounds during training to promote robustness. These include constructing provable lower bounds on the norm of the input manipulation required to change classifier decisions based upon the network architecture~\cite{hein2017formal}; modifying the loss associated with logits different than the ground-truth class~\cite{tsuzuku2018lipschitz}; and GloRoNets, which add an additional logit corresponding to the predicted class at a point~\cite{leino2021globally}. While these techniques can be an order of magnitude faster than randomized smoothing, they are both less flexible---in terms of the architectures they support---and often produce smaller certifications than randomized smoothing.~\cite{leino2021globally}. 

\section{Certifications for Model Security}

Whether presented as a `certified guarantee' or a `certified defense', that it is a `certification' heavily implies an absolute improvement to model security. 
This impression is driven not just by the naming of these techniques, but also how they are described. After all, these are techniques that are guaranteed to apply in all circumstances, irrespective of the attacker's behavior. To both academic and non-academic readers  who are even passingly familiar with the security risks associated with adversarial examples, such properties are incredibly appealing. 

However, it must be emphasized that certified defenses do not operate in the same manner as a traditional defense. While a traditional defense ideally increases the difficulty of performing an attack, a certification only measures the distance to the nearest class-flipping example. In the literature this is typically framed as the distance to the nearest possible adversarial example, however this is not strictly true for deployed models, as \emph{adversarial examples can also be certified}~\cite{cullen2023exploiting}. That a certification is the distance to the nearest possible adversarial example is only true under the settings of many academic papers, in which oracle level knowledge of the true class is presumed. 

This clear disparity between how certifications may be perceived, and what they actually produce presents a security risk that can potentially be exploited by motivated attackers. After all, if a model deployer is confident that their model is certifiably robust against adversarial examples, there is potentially no need to implement any other security measures. This is especially worrisome when certification mechanisms are inherently limited to specific types of threat models---for example, geometric attacks (rotational or translational) are unlikely to be covered by traditional $\ell_p$ based threat models~\cite{xiao2018spatially, dumont2018robustness}.

In practice, as the following theorem argues, the only information provided by the certificate is the distance to the nearest potential class-flipping example, rather than providing any information regarding if the sample has been attacked or not. If a point is correctly predicted, then this distance may be the distance to the nearest adversarial example, or to the true semantic class boundary. However, if the point is an adversarial example whose class expectation is large enough to produce a certification, then the certification is the distance to the true class.

If we know that any potential attacker is $\epsilon$ bounded within the $\ell_p$ norm that we have been able to certify, then the guarantee will ensure that the class prediction will remain constant for these attacks. However, this does not guarantee  that the prediction is correct, nor that it has not been the subject of an attack. While it may be true that certifiable adversarial examples may produce smaller certifications, due to the inherent proximity of adversarial examples to decision boundaries, this is only a heuristic, with no theoretical backing. While this observation may allow certifications to be used to stack-rank risk using certifications in a comparative fashion, we would argue that the only reliable , actionable information that a certification technique may currently provide is the absence of a certification.

\begin{theorem}
    A certification of size $\epsilon$ associated with the input $\mathbf{x}$ to a model $f$ could correspond to either a certification of the correct class, that is representative of the semantic space that the sample exists within; or a certification of an incorrect class, one which is not representative of the semantic space a sample exists within. Thus the existence of a certification does not intrinsically provide any information regarding if the sample $\mathbf{x}$ has been attacked or not.   
\end{theorem}

As an example of this, consider a stochastic, location invariant classifier, that produces a fixed expectation of $0.75$ and a constant class prediction across all $\mathbf{x} \in \mathcal{S}$. While this classifier will certify all points, the classifier will have low accuracy, and the certified will not provide any actionable information. While this point may appear obvious, it underpins the inherent contradiction between how certifications are presented---as a security guarantee---and how they operate in actuality. 

An additional consequence of the above theorem is that having access to a certification provides an attacker additional information regarding where adversarial examples may or may not be~\cite{cullen2023exploiting}. This allows an attack to be guided not just by gradients, but by knowledge of where adversarial examples can and cannot exist. Thus, if the attacker has access to the certifications, then they have an information advantage relative to an uncertified model. Thus it can be argued that \emph{employing certification mechanisms may compromise AI security}~\cite{cullen2023exploiting}. 

\subsection{Employing Certifications}

Given these observations, there is a clear need for the certification research community to acknowledge the limitations inherent to certifications, and to reflect on how the framing of these techniques may drive misapprehensions about the levels of security provided. At the most basic level, this should include emphasizing that certifications should only be accessible to those who have trusted access to the model. Throughout this paper, we will explore potential research directions relating to both this and other issues through a series of open questions. 

\begin{oq}
How best should certifications be employed to enhance model security? 
\end{oq}
Based upon our discussions to this point, we can treat a certification as a heuristic measure of how likely it is that a sample may have been manipulated. However, thinking about a certification in isolation also potentially minimizes how information security is practiced. 

To take a more systematic perspective, consider organizational security as a composition of operations, that could include rules, algorithmic screening, human operators, and more. If a model produces a radius, how would that information be best served by other components of that process flow? Should the information of the certification be propagated through to subsequent tasks (or even back to earlier operations)? How could a certification be incorporated into a multifaceted assessment of risk, for both individual samples and for collective sets? Can certifications be informed by measures of risk at preceding steps of the model pipeline? 

These questions may seem vague, but it is crucial to think about how techniques designed for mitigating risk---as certifications are designed to do---may exist in the context of real-world risk management frameworks. Both ISO/IEC 27001:2022 and the NIST AI Risk Management Framework~\cite{ISO27001_2022, NIST_AI_RMF_2022} treat AI systems governance as something that requires continuous, active, multifaceted risk monitoring and assessment. For organisations, conforming to such information security controls is crucial not just for managing their own risk, but for aligning with legal and auditing expectations. In the case of smaller organisations, simply recording certifications may be enough to satisfy auditing requirements, but more complicated security apparatus will require a more nuanced perspective to be taken. 

The challenge with attempting to answer questions like these in an academic context is that they do not align well with the tools that we have at our disposal. We do not have easy access to real-world information security risk frameworks. And even if we did, any testing we performed would likely produce results that were specific to particular organizations. This is not to say that these problems are not able to be studied within an academic context. In fact, facets of this problem space can be seen in the fields of mathematical risk management, human-in-the-loop computing, human computer interaction, game theory~\cite{zhou2019survey, sun2023adversarial, cullen2024game, adams2025suboptimality}, and psychology. This suggests that working towards a more holistic view of certified robustness will require multidisciplinary research expertise. 

Taking such a perspective is critical to avoid certifications becoming more security-theater than actual security. As has been noted in the differential privacy community, the inherent trade off between user privacy and utility in differentially private systems creates a tension that has the potential to lead system creators to minimize transparency. Doing this has the potential to convert privacy guarantees into advertizing material and window dressing, that provides only the appearance of positive user benefit~\cite{khare2009privacy}. In response to this, recent observational studies have begun to consider both how expectations of privacy shape user habits, and how clarifying private mechanisms can induce confidence in system privacy \cite{xiong2020towards, smart2022understanding}. If certification schemes are to be considered as similarly important for demonstrating model security, then it is important to both consider and study how the framing of these mechanisms affects user expectations.

\begin{oq}
    What is required for certifications to be practically deployed for end users?
\end{oq}

While works examining randomized smoothing, IBP, and global Lipschitz-style certifications often highlight their relative benefits, the level of detail provided is typically insufficient for end-users to assess whether an approach suits their needs. This is particularly true when user requirements span factors such as resource demands, ease of deployment, and certification performance on datasets relevant to their use cases. We believe it is crucial for researchers to develop a shared framework for analyzing certification schemes, offering more contextual information about their performance.

Typically, certification works allude to their employed computational resources, which is sometimes supplemented with a discussion of the total computational time required. However, in practice comparisons between the resource demands imposed by different techniques are rare, and yet these very comparisons are crucial for determining the suitability of these schemes for end users. This is especially so when production environments may not share the same bottlenecks as research systems, which may lead to differing perspectives on how computational costs would be perceived. 

In practice, we contend that practitioners should better structure their comparisons in terms of the resource requirements, the computational time required, and the level of achievable parallelism. While it is simple to state this as a necessity, in practice these comparisons are complicated by the stark methodological differences between the core techniques. To take randomized smoothing as an example---the large number of draws required to construct the expectations may appear to be numerically expensive. However, in practice this task of repeated model draws is embarassingly parallel, can be split over arbitrarily many GPUs, and only requires as much memory as is required to hold the model. In contrast interval bound propagation typically only requires a single pass to establish a certification. However implementing this requires both significant amounts of computational time and GPU memory to construct the certification, which intrinsically limits the size of models that can be certified.

While there is a paucity of comparisons between the different certification frameworks, the International Verification of Neural Networks Competition's (VNN-COMP) comparisons between Interval Bound style certification mechanisms serves as a proof of concept for how these comparisons could be performed~\cite{muller2022third, brix2023fourth}. VNN-COMP builds comparisons between the success rates and running times, while controlling for run-time related issues by providing a shared codebase and prescribed computational environments, demonstrating that it is possible to begin to construct broader comparisons. However it is crucial to emphasize that the VNN-COMP comparisons only exist for IBP style certifications, and do not consider randomized smoothing or Lipschitz approaches.

The net result of constructing more rigorous estimates on computational cost will likely require a broader set of experiments than those typically performed in certification papers---especially with regard to the impact of different model sizes. However, it is important to stress that such an analysis should not be strictly rooted in trying to demonstrate the superiority of a technique, but it should rather be focused upon delving into the properties of the technique. 

In the absence of clear practices for deploying certification schemes, research on computational cost should aim not to prove the superiority of any technique, but to provide knowledge that helps practitioners decide whether to use a certification framework. While this may be challenging given publishing conventions focused on state-of-the-art improvements, it could open new opportunities for comparing different certification schemas.

\begin{oq}
    How do we test certification schemes in a manner that reflects real world use cases?
\end{oq}

Establishing certification performance on key reference datasets like MNIST \citep{lecun1998gradient}, CIFAR-$10$ \citep{krizhevsky2009learning}, and the Large Scale Visual Recognition Challenge variant of ImageNet \citep{deng2009imagenet, russakovsky2015imagenet} are important tools for validating research works. However, the semantic properties of these datasets, and their diversity---or, more precisely, their lack thereof---limits the ability to transfer these results to other datasets of interest. This has even been shown to extend to datasets in different contexts to those in which the datasets were originally sampled, due to geographic and cultural biases that are driven by the very mechanisms through which these datasets were originally constructed~\cite{buolamwini2018gender, celis2020implicit, karkkainen2021fairface, mandal2021dataset}. While some task-specific works have begun to consider broader views on certification datasets~\cite{dvijotham2020framework, korzh2024certification}, there clearly exists significant space for broadening the scope of how these systems are evaluated, to better demonstrate and understand utility. 

As noted by~\citet{cullen2023its}, the performance properties of different certification techniques can vary based upon the distribution of points within what they describe as the simplex of potential output spaces. As it is likely that datasets of interest may not share the same properties as those employed within academic research, it is important that we broaden our appreciation of what exactly state of the art is, and how techniques can be selected to maximize utility for specific tasks. 

Beyond this, while improving the size of certified guarantees will always be important, it is also crucial that users are supported with the information to contextualize said guarantees. After all, a certification with an $\ell_p$ size of $2$ (for some $p$) likely does not intrinsically convey enough knowledge to understand the risk associated with a sample from an arbitrary dataset being attacked---for it may be that all samples are clustered within a distance of $2$ of the sample point, or there may not be a single other clean sample within this radius. Thus, for these systems to have real world applicability and interpretability, techniques to contextualize certification sizes are crucial. 

A source of inspiration for improving the quality of testing within the certification literature is Instance Space Analysis~\cite{smith2010understanding, munoz2018instance}, which can be used to create a representative footprint of where samples may exist. This data-driven approach allows practitioners to both quantify how much coverage a dataset provides over potential input space, but also can guide the generation of new datasets. Drawing from such approaches may be useful to better understand the factors that drive certification performance in datasets that do not resemble the community's typical reference datasets. 

\section{Coverage}

To extend upon the preceding discussions of improving certifications to enhance the concept of model security, we now turn to more practical considerations.

\begin{oq}
    How do we improve the quality of guarantees provided by certifications?
\end{oq}

Intuitively, the size associated with a certification is directly correlated with its applicability, with larger regions of coverage providing more general guarantees, and more security. While true, it is also important to note that current certifications are often small enough to not obviate the existence of imperceptibly small adversarial perturbations. Thus increasing the size of certifications will inherently decrease the risk of attack~\cite{gilmer2018motivating}. 

This perspective on certification size being a direct measure of risk is challenged by geometric perturbations, where the  $\ell_p$ distance may not reflect the level of difficulty in either constructing or detecting a manipulation. Moreover, any $\ell_p$ certification can be negated if the attacker can operate in some space $\ell_q : p \neq q$. While there is overlap between the regions of coverage provided by differing $\ell_p$ spaces, the potential for shifting the attack norm may introduce opportunities for the attacker to exploit. 

To understand the implications of this, it is important to remember that for an attacker to be successful, they only need to find a single working adversarial example. By contrast, a defender ideally must prevent all adversarial examples from being passed through the model. In a certification context, consider a scheme that produces an $\ell_p$-norm ball of size $r_p$. As indicated by Theorem~\ref{theorem:bounding}, if $q < p$ no $\ell_q$-norm adversarial examples exist with size $r_q < r_p$. However, if $q > p$, \emph{smaller} $\ell_q$ \emph{norm adversarial examples may exist}! Thus mismatches between attacker and certification norms can potentially induce an unfounded sense of security. This is especially so when it is possible that the certification norm potentially does not align with the capacity for these adversarial examples to be detected. 

\begin{theorem}\label{theorem:bounding}
    Consider a $\ell_p$-certification out to a distance of $r_p$. Potential adversarial attacks for an attacker operating with an $\ell_q$-norm attack exist for $r_q > \min\left( d^{1/q - 1/p}, 1 \right) r_p$.
\end{theorem}

\begin{proof}
    When $q < p$, the two regions of certification intersect at $r_q = r_p$, and thus there exists some points in the region $r_q > r_p$ that admit potential adversarial examples. 

    When $q > p$, the regions of certification intersect at $\frac{r_q}{d^{1/q}} = \frac{r_p}{d^{1/p}}$, for a $d$-dimensional space. Thus there exists points at distances $r_q > d^{1/q - 1/p} r_p$ which are not covered by the $\ell_p$ region of certification. 
\end{proof}

While Theorem~\ref{theorem:bounding} does present a mechanism for translation from $\ell_p$ certifications to $\ell_q$ threat models, ideally we should be considering how to optimize certifications for the threat model of interest, as can be seen in recent works that have begun to generalize certifications away from $\ell_2$ threat models~\cite{yang2020randomized,huang2023rsdel}. We should also explore moving past individual threat models to instead consider maximizing certification coverage. However, this will require us to fundamentally change how we assess certification performance. While prior works have demonstrated that it is possible to ensemble certifications \cite{cullen2023its}, their approach was still rooted in an $\ell_2$ space. Ultimately maximizing coverage may require new certification mechanisms in other $\ell_p$ spaces and even non-$\ell_p$ threat models such as edit distance for sequence classifiers~\cite{huang2023rsdel}. It may also require balancing the costs of performing multiple certifications, and the added utility provided by such a layering. This leads to an additional question, regarding how new certifications can be constructed. 

\begin{oq}
    How can we build certification mechanisms that can be generalized to a broader set of model types?
\end{oq}

To this point, while we have attempted to be general in our consideration of certifications, there has been an inherent bias towards the robustness of classifiers, and classifier-like systems. This bias reflects that of the overall certification space, which is heavily weighted towards works considering classifiers under $\ell_2$-norm bounded (or $\ell_p$) threat models. While ensuring classifier performance is important, there is no guarantee that the AI systems that we will most heavily rely upon in the future will be similar, nor that the risks of adversarial manipulation will be limited to such classifiers~\cite{mangal2023certifying}. While recent works have begun to consider how certified robustness can be generalized to frameworks like reinforcement learning~\cite{lutjens2020certified, kumar2021policy, mu2023certified, wu2022copa}, there still remains significant potential for expanding the scope of problems considered through certifications. 

\section{Secure Development}

Finally, it is critical to consider how certification techniques can be developed into secure models and certification code implementations. After all, the security guarantees provided by certifications will be for naught if they cannot be incorporated into deployed code. 

\begin{oq}
    How do we incentivize the development of more secure code?
\end{oq}

To date, research projects on certification have remained at low levels of technology readiness. While this is understandable given the level of maturity of the field, if there is to be adoption of these systems, and impact outside of research, it is clear that significant care must be taken to implement secure code. This goal would be supported by systems that are open and transparent, to both enable audits~\cite{ding2017privacy} and enhance user confidence in system performance. Extensible implementations would also facilitate adaptation to a range of applications. 

A source of inspiration is OpenDP~\cite{openDP2020, gaboardi2020programming} which has generalized custom, task-specific tools to a domain-agnostic framework for enhancing the privacy guarantees associated with data access principles. OpenDP's development process is an exemplar of how security and accessibility can be balanced through an open, application driven development process, that is built upon a rigorous code- and proof-checking process. In doing so, OpenDP has managed to rapidly gain significant buy-in from both researchers and developers~\cite{lokna2023group}. 

It would be possible for a similar set of processes to be employed for the development of more production-ready certification code-bases. The advantages of such an approach would not just be useful for potential deployments of certification frameworks, but would also allow for different schemes to be more readily tested against each other. Such comparisons would only benefit the development of the field. Moreover, a shared research framework could also help induce a greater sense of confidence in new works, as it would be clear that they were operating upon the same framework employed by earlier techniques. 

Additional lessons can also be drawn from the specific development pathways of cryptographic and differential privacy implementations, and how these pathways have lead to new research developments. Rather than directly implementing state-of-the-art research, OpenDP and crytrographic protocols have often draw inspiration from these systems while focusing upon ensuring that they work to a set of assumptions that allows for the develop of reliably robust, testable, and verifiable systems. This focus upon how systems are employed in practice, in contrast to the more standard academic assumptions that had existed in prior works then in turn lead to new areas of research interest in areas like rounding and floating point issues in Differential Privacy correctly~\cite{mironov2012significance, balcer2017differential, jin2022we}, the development of side-channel~\cite{jin2022we} and floating-point attacks~\cite{jin2024getting}, and explorations of the impact of privacy budgets~\cite{jin2024elephants}. These examples provide a clear precedent for how examining security problems with an eye to how they will be deployed within the real world can pay both practical and research dividends. 

\section{Alternative Views}

While the preceding content argues that certified robustness schemes share significant open weaknesses, a counterfactual perspective would be that any improvements provided by a certification still enhances model security, even if these improvements do not provide complete security. While this is true, as we have argued within this paper, we believe there is a high likelihood that these schemas will result in security theatre, rather than security. The likelihood of this is heavily driven by the presentation of certifications as a guarantee of robustness. When certifications are marketed as definitive proof of a model’s resilience to adversarial examples, it creates the illusion of a level of protection that may not truly exist. Such a false sense of security could divert attention from more robust, ongoing security measures and research, potentially leading users to neglect further model improvements or defensive strategies. Moreover, incorporating certifications may also foster complacency in model development.

It could also be argued that these expectations are not the responsibility of the certification community, and that these systems are being developed for technical users at this stage, with future developments taking care of ease of use and broader adoption. While basic research is undeniably important, these benefits do not negate concerns for how certifications are being communicated. For technical users, it is essential to understand the limits of certified robustness guarantees and the broader implications of these systems, particularly in industries that rely on AI models. If certifications are not accompanied by clear explanations of their limitations and scope, there is a risk that non-technical stakeholders---who may not fully grasp the underlying complexities---could misinterpret or overestimate the significance of these certifications. This disconnect could lead to misuse or overreliance on certifications, which in turn may hinder the development of more comprehensive security strategies. In the long term, as these systems become more accessible and widespread, there will be an increased need for transparency and clear communication about what certified robustness can and cannot achieve. Without this, the potential for security theater remains a significant concern.

\section{Conclusions}

The AI security community has consistently been producing research that moves the needle in terms of our understanding of risks facing models, and the strategies that can be employed to mitigate said risks. However, for all of the research insights gained, there is a strong case to be made that the current literature has not yet bridged the gap from practical promise to deployable applications. This is particularly true for certified robustness, given the expectations of users searching for a guaranteed solution to the risk of adversarial manipulation. 

Indeed, the very framing of certified robustness as providing guarantees of adversarial resistance can create a false sense of security, and more broadly, an alignment issue between how users would likely perceive these systems and their actual performance. Given this, within this paper we argue that additional care must be taken to ensure that these certifications are presented in a manner that prevents their being misinterpreted as blanket protections, rather than constrained assurances tied to specific threat models, perturbation bounds, and data distributions. Such misunderstanding can lead to adverse security outcomes, especially in real-world deployment scenarios where defenders do not have access to oracle information, and where adversaries do not necessarily conform to the narrow theoretical bounds underpinning certification techniques. 

This work is not just an excoriation of current research practices. We also argued that considering how these systems are used and perceived can inspire new, interesting research directions. This work would extend beyond the current remit of certification research---which is broadly focused on improving bulk metrics on reference data sets---into exploring the how and why of certification performance, the factors that incentivize the development of secure code, the applicability of certifications to real world threat models, and how they should be communicated to stakeholders. These rich new veins of research questions have the potential to significantly improve the safety of deployed AI systems.

\section*{Impact Statement}

This work takes the position that for all its promise, certified robustness still has a long way to go before being ready for wide-spread deployement as a real-world method of securing AI. While this position may superficially appear to be pessimistic, we believe that constructive discussion about the current state of certification research can help reveal new, productive research directions. If the research community can progress solutions, we believe that certifications can provide a tangible impact on AI security, especially where AI is deployed in high-risk and high stakes contexts, for positive societal impact.

\section*{Acknowledgments}

This work was supported by the Australian Defence Science and Technology (DST) Group via the Advanced Strategic Capabilities Accelerator (ASCA) program.

\bibliography{./icml2024_conference_bib}
\bibliographystyle{./iclr2023_conference}

\end{document}

%% file: math_commands.tex
\usepackage{amsmath,amsfonts,bm}

\def\eqref#1{equation~\ref{#1}}

\def\1{\bm{1}}

\DeclareMathAlphabet{\mathsfit}{\encodingdefault}{\sfdefault}{m}{sl}
\SetMathAlphabet{\mathsfit}{bold}{\encodingdefault}{\sfdefault}{bx}{n}

\DeclareMathOperator*{\argmax}{arg\,max}
\DeclareMathOperator*{\argmin}{arg\,min}